\documentclass[letterpaper, 10 pt, conference]{ieeeconf}  

\IEEEoverridecommandlockouts                              
\usepackage{graphicx}

\usepackage{amsmath}
\usepackage{amssymb}
\usepackage{amsfonts}
\usepackage{mathtools}
\usepackage{MnSymbol}

\usepackage{subcaption}
\usepackage{amsthm}
\usepackage[usenames,dvipsnames,svgnames,table]{xcolor}
\usepackage{float}
\floatstyle{plaintop}
\restylefloat{table}
\usepackage{cite}
\usepackage[normalem]{ulem}
\usepackage{hyperref} 

\usepackage[vlined,ruled,linesnumbered]{algorithm2e}
\theoremstyle{plain}
\newtheorem{theorem}{Theorem}
\newtheorem{definition}{Definition}

\newtheorem{problem}{Problem}

\newtheorem{remark}{Remark}
\newtheorem{lemma}{Lemma}

\DeclareMathOperator{\luntil}{{\mathcal{U}}}
\newcommand{\levent}{\lozenge}
\newcommand{\lalways}{\square}

\newtheorem{prop}{Proposition}

\title{\LARGE \bf
Control Barrier Function for Linearizable Systems with High Relative Degrees from Signal Temporal Logics: A Reference Governor Approach
}

\author{Kaier Liang, Mingyu Cai, and Cristian-Ioan Vasile
\thanks{Kaier Liang and Cristian-Ioan Vasile  are with the Mechanical Engineering and Mechanics Department at Lehigh University, PA, USA: {\tt\small \{kal221, cvasile\}@lehigh.edu}}
\thanks{Mingyu Cai is with the Department of Mechanical Engineering at University of California Riverside, CA, USA: {\tt\small mingyu.cai@ucr.edu}}
}     
\renewcommand{\baselinestretch}{0.95}

\begin{document}
\maketitle
\thispagestyle{empty}
\pagestyle{empty}

\begin{abstract}
This paper considers the safety-critical navigation problem with Signal Temporal Logic (STL) tasks.
We developed an explicit reference governor-guided control barrier function (ERG-guided CBF) method that enables the application of first-order CBFs to high-order linearizable systems. 
This method significantly reduces the conservativeness of the existing CBF approaches for high-order systems. Furthermore, our framework provides safety-critical guarantees in the sense of obstacle avoidance by constructing the margin of safety and updating direction of safe evolution in the agent's state space.
To improve control performance and enhance STL satisfaction, we employ efficient gradient-based methods for iteratively learning optimal parameters of ERG-guided CBF. We validate the algorithm through
both high-order linear and nonlinear systems. A video demonstration can be found on: \url{https://youtu.be/ZRmsA2FeFR4}
\end{abstract}

\section{Introduction}
\label{sec:introduction}
Control design for safety-critical systems subject to state constraints has become an important research direction in robotic applications. Furthermore, robots are frequently tasked with complex assignments, which can be expressed in Signal Temporal Logic (STL) \cite{maler2004monitoring}, a formal language interpreted over continuous-time signals used to formulate tasks with time windows and deadlines. Recent research employs STL to learn rules of autonomous control systems from data for interpretable reasoning~\cite{leung2023backpropagation, li2023learning, aasi2023time}.

Control Barrier Functions (CBFs)~\cite{ames2019control} have recently drawn considerable interest for safety-critical applications. By constructing a forward invariant safe set via the barrier functions and solving for the control input using quadratic programming, CBFs ensure that the system remains within the safe set. CBFs provide a highly effective tool for designing provably safe controllers that are computationally efficient~\cite{ames2019control}. In~\cite{lindemann2018control}, time-varying CBFs were used to enforce a fragment of STL specifications for first-order systems. For systems with a relative degree greater than one, \cite{xiao2021high} introduces High-Order Control Barrier Functions (HOCBFs). However, HOCBFs are typically conservative, which could render the problem infeasible when the safe set is restricted~\cite{xiao2023barriernet}.

Moreover, constructing CBFs involves hand-designing their structures and fine-tuning their parameters with significant impacts on performance. 
Learning CBF parameters from expert demonstrations using an optimization-based approach was explored in~\cite{robey2020learning}.
A differentiable learning framework for class K functions for exponential CBFs was developed in~\cite{ma2022learning} that facilitates generalization to novel environments.

The neural controller BarrierNet~\cite{xiao2023barriernet} was proposed to reduce conservativeness for HOCBFs. It is used to learn the parameters of STL specifications~\cite{liu2023learning} to improve performance.
Another work~\cite{molnar2021model} proposed first-order CBFs for a safe set in velocity space and applied for velocity tracking using the control Lyapunov functions (CLF) to ensure safety-critical navigation. 
However, the approach requires designing CLF parameters to enable sufficiently fast tracking performance.
 

Model Predictive Control (MPC)~\cite{rawlings2000tutorial} is a well-established method to address constraint control. Control input in MPC is obtained from an optimization problem over a fixed time horizon at each time step. Its ability to handle various constraints has been proven successful in numerous real-world applications~\cite{mayne2014model}. However, the heavy reliance of MPC on online optimization often results in greater computational burden. The integration of STL and MPC is discussed in~\cite{raman2014model,sadraddini2015robust}, which involves the construction of demanding mixed-integer linear programs.

To overcome these challengers, the explicit reference governor (ERG) was introduced in \cite{nicotra2018explicit}. The  methodology first constructs a dynamic safety margin (DSM) based on the zero order safety set, followed by defining a navigation field (NV) to indicate the direction of adjustment for the reference governor. However, the construction of DSM and NV can be complex, often depending on specific models and constraints. In~\cite{li2023governor}, the authors construct the DSM for feedback-linearizable control-affine nonlinear systems for safe navigation using the barrier function, while the governor update direction is defined by projection to a predefined reference.


 This paper proposes ERG-guided CBFs for STL satisfaction with safety guarantees, which enables the application of first-order CBFs to feedback linearizable systems with high relative degrees. We employ a first-order linear system as a reference governor and construct the dynamic safety margin based on~\cite{li2023governor}. Then the navigation field is developed using time-varying CBFs to ensure that the high-order system navigates safely through narrow passages and maximizes the satisfaction of STL tasks. In addition, we apply a gradient-based method for auto-tuning the parameters of feedback control gain to improve the performance of satisfying STL specifications.

\section{Preliminary}

Consider the nonlinear control affine system:
\begin{equation}
\dot{x} = \boldsymbol{f}(x)+\boldsymbol{g}(x)u
\label{eq: nonlinear_sys}
\end{equation}
where ${x} \in \mathbb{R}^n$ is the state of the system and ${u}\in \mathcal{U}\subset \mathbb{R}^m$ is the control input, $\boldsymbol{f}: \mathbb{R}^{n}\rightarrow\mathbb{R}^{n}$ and $\boldsymbol{g}: \mathbb{R}^{n}\rightarrow\mathbb{R}^{n\times m}$ are locally Lipschitz
continuous functions, $\mathcal{U}$ is a box constraint, i.e., $u_{\text{min}}\leq u\leq u_{\text{max}}$.

We assume system~\eqref{eq: nonlinear_sys} is feedback linearizable~\cite{khalil2015nonlinear} and results in a linear time-invariant dynamical system:%
\begin{equation}
\begin{aligned}
    \dot{{x}}&={Ax}+{Bu}\\
    y &= Cx
\label{eq: linear_sys}
\end{aligned}
\end{equation}
where $y \in \mathbb{R}^p$ is the system output.

\subsection{Signal Temporal Logic (STL)}

Signal Temporal Logic~\cite{maler2004monitoring} is a predicate logic defined over signals ${x} : \mathbb{R}^+ \to \mathbb{R}^n$. 
Let $\mu ::= h(x) \geq 0$ represent a predicate, where $h: \mathbb{R}^n \rightarrow \mathbb{R}$ is an evaluation function of a state $x\in \mathbb{R}^n$.

We consider the following fragment of STL:
\begin{equation*}
    \begin{aligned} & \psi::=\top|\mu| \neg \mu \mid \psi_1 \land \psi_2 \\ & \phi::=G_{[a, b]} \psi\left|F_{[a, b]} \psi\right| \psi_1 U_{[a, b]} \psi_2 \mid \phi_1 \land \phi_2\end{aligned}
\end{equation*}

where $\psi, \phi_1, \phi_2$ are STL formulas.
The temporal eventually, always, and until operators with time interval $I$ are $\levent_I$, $\lalways_I$ and $\luntil_I$, respectively.

The semantics of STL are evaluated over trajectories $x(t)$:
$$
\begin{aligned}
& ({x}, t) \models \mu &&\Leftrightarrow h({x}(t)) \geq 0 \\
& ({x}, t) \models \neg \phi &&\Leftrightarrow \neg(({x}, t) \models \phi) \\
& ({x}, t) \models \phi_1 \land \phi_2 &&\Leftrightarrow({x}, t) \models \phi_1 \land({x}, t) \models \phi_2 \\
& ({x}, t) \models \phi_1 \luntil_I \phi_2 &&\Leftrightarrow \exists t_1 \in t+I \text { s.t. }\left({x}, t_1\right) \models \phi_2 \\
&&& \qquad \land \forall t_2 \in\left[t, t_1\right],\left({x}, t_2\right) \models \phi_1 \\
& ({x}, t) \models \levent_I \phi \quad &&\Leftrightarrow \exists t_1 \in t+I \text { s.t. }\left({x}, t_1\right) \models \phi \\
& ({x}, t) \models \lalways_I \phi \quad &&\Leftrightarrow \forall t_1 \in t+I,\left({x}, t_1\right) \models \phi .
\end{aligned}
$$

\subsection{Time-varying CBF for STL}
CBFs are often used to design safe controllers by ensuring that a safe set is forward invariant: a system that starts in the safe set stays in the safe set~\cite{ames2019control}. The controller is obtained through an efficient quadratic program (QP).
In~\cite{lindemann2018control}, the time-varying CBF $b(x, t)$ given by
%
{\small\begin{equation}
\label{eq: cbf}
\begin{aligned}
&\text{min} \quad u^\top Q u\\
    &\sup _{{u} \in \mathcal{U}} \frac{\partial {b}({x}, t)^\top}{\partial {x}}(f({x})+g({x}) {u})+\frac{\partial {b}({x}, t)}{\partial t} \geq-\alpha({b}({x}, t))
\end{aligned}
\end{equation}
}%
are used to ensure the satisfaction of formulas from a fragment of STL, where $\alpha$ is a class $K$ function. 
If~\eqref{eq: cbf} holds for all $x(t)$ and $b(x(0), 0) >0$ then the system is positively forward invariant, i.e., $b(x(t), t) \geq 0$ $\forall t$.

In~\cite{lindemann2018control}, CBF is designed according to predicates of the STL formula.
These CBFs are combined to achieve complex tasks involving conjunction and temporal operators.
For example, the CBF for $\phi = \phi_1 \land \phi_2 \land \ldots \land \phi_n$ uses an approximation of the minimum and is given by
%
\begin{equation}\label{eq: cbf_and}
    b_{\phi}(x(t),t) = -\ln \sum_{i=1}^n \exp(-b_{\phi_i}(x(t),t)).
\end{equation}
Therefore if $b_{\phi}(x(t),t) >0$, then $\forall i,$ $b_{\phi_i}(x(t),t) > 0 $.

\subsection{Explicit Reference Governor}

The explicit reference governor (ERG)~\cite{nicotra2018explicit} is an efficient control design technique for constraint handling. Given the dynamics in~\eqref{eq: linear_sys}, a defined desired reference $r(t): \mathbb{R} \rightarrow \mathbb{R}^p$, and the constraints function $c(x(t),r(t)): \mathbb{R}^n\times \mathbb{R}^p\rightarrow \mathbb{R}$ that requires:
\begin{equation}
    c(x(t),r(t)) \geq 0
\label{eq: constr}
\end{equation}
The ERG framework generates the auxiliary reference $g(t):  \mathbb{R}\rightarrow \mathbb{R}^p$ such that the constraint $c(x(t), g(t)) \geq 0$ is satisfied for all $t\geq 0$. The auxiliary reference is updated as:
\begin{equation}
    \dot g = \Delta(x,g)\rho(r, g),
\label{eq: au_ref}
\end{equation}%
where $\Delta(x,g) \in \mathbb{R}$ is called dynamic safety margin (DSM) and $\rho{(r, g)} \in \mathbb{R}^p$ is called the navigation field (NV).

Let $\bar{x}_g: \mathbb{R}^p \rightarrow \mathbb{R}^n$ be a continuous mapping that denotes a corresponding desired state to $x$ associated with reference $g$.  

\begin{definition}\label{def: DSM}
  ~\cite{nicotra2018explicit} For a fixed reference $g$, a continuous function $\Delta: \mathbb{R}^n \times \mathbb{R}^p \rightarrow \mathbb{R}$ is a dynamic safety margin if
    \begin{enumerate}
        \item $\Delta(x, g)>0 \Rightarrow c(x(t),g)>0$, for all $t \geq 0$;
        \item $\Delta(x, g) \geq 0 \Rightarrow c(x(t),g) \geq 0$, for all $t \geq 0$;
        \item $\Delta(x, g)=0 \Rightarrow \Delta(x(t),g) \geq 0$, for all $t \geq 0$;
        \item For all $\delta>0$, there exists $\epsilon>0$ such that $c\left(\bar{x}_g, g\right) \geq \delta \Rightarrow$ $\Delta\left(\bar{x}_g, g\right) \geq \epsilon$.
    \end{enumerate}
\end{definition}

The dynamic safety margin guarantees the satisfaction of constraints. Specifically, larger values of DSM indicate the system is safer with respect to the constraints. $\delta$ can be seen as the static safety margin. Next, the navigation field specifies the direction of system updates for safe tracking.

\begin{definition}\label{def: NV}~\cite{nicotra2018explicit}A piecewise continuous function $\rho(r, g): \mathbb{R}^p \times \mathbb{R}^p \rightarrow \mathbb{R}^p$ is a navigation field if for any initial condition $g(0)$ satisfying the constraint~\eqref{eq: constr}, the system
\begin{equation}
\label{eq: nv}
    \dot g = \rho(r,g)
\end{equation}
is such that
\begin{enumerate}
    \item $\sup _{(r, g) \in H}\|\rho(r, g)\|$ is finite for each compact set $H$.
    \item For any piecewise continuous reference $r(t) \in \mathbb{R}^p$, the result $g(t)$ satisfies $c\left(\bar{x}_r, g(t)\right) \geq \delta$.
    \item For any constant reference $r$ such that $c\left(\bar{x}_r,r\right) \geq \delta$, the equilibrium point $g=r$ is asymptotically stable and admits $\left\{g: c\left(\bar{x}_g,g\right) \geq \delta\right\}$ as a basin of attraction.
\end{enumerate}
\end{definition}
The navigation field characterizes the asymptotic stability while admitting the reference $\{g: c(x,g) \geq \delta\}$ as a basin of attraction.

\begin{theorem} ~\cite{nicotra2018explicit}
\label{thm: erg}
Consider the prestabilized system in~\eqref{eq: linear_sys} and constraint in~\eqref{eq: constr}. Given the initial condition $ x(0), g(0)$ at $t=0$ satisfying $c(x(0), g(0)) > 0$. The update law of the governor in~\eqref{eq: au_ref} has the properties:
\begin{enumerate}
    \item For any piecewise continuous reference signal $r(t) \in \mathbb{R}^p$, constraints in~\eqref{eq: constr} are never violated.
    \item For any constant reference $r$ such that $c\left(\bar x_r,r\right) \geq \delta$, the equilibrium point $\bar x_r$ is asymptotically stable and admits $\left\{(x, g): c\left(\bar x_g,g\right) \geq \delta, \Delta(x, g) \geq 0\right\}$ as a basin of attraction.
\end{enumerate}
\end{theorem}
The proof is given in~\cite{nicotra2018explicit}. By properly defining the dynamic safety margin and navigation field, the ERG can be used to generate the auxiliary reference to ensure safe tracking for a Prestablized system.

\section{Problem Formulation}
Consider the dynamic system in~\eqref{eq: linear_sys}.
We define an obstacle-free open set $\mathcal{F} \subset \mathbb{R}^p$ and a closed obstacle set $\mathcal{O}:=\mathbb{R}^p \backslash \mathcal{F}$. The safe state set is $\mathcal{F}_x = \{x\mid y=Cx\in\mathcal{F}\}$.
The objective for the system is to comply with an STL formula $\phi_{stl}$ while simultaneously preserving safety
\begin{equation}
\begin{aligned}
& \dot{{x}}={Ax}+{Bu}\\
&y = Cx\\
    &\quad \text{s.t. } y(t) \in \mathcal{F}.\\
    &\quad (y,t) \models \phi_{stl} 
\end{aligned}
\label{eq: stl_obs}
\end{equation}

\textbf{Motivation. }
The relative degree of a differentiable function \(b(x)\) is the number of times it must be differentiated along the dynamics of system~\eqref{eq: nonlinear_sys} until the control input \(u\) explicitly appears in the corresponding derivative. Formally, the relative degree \(r\in \mathbb{Z}^+\) is such that 
$L_{\boldsymbol{g}}L_{\boldsymbol{f}}^{r}b(x) \neq 0$ and $L_{\boldsymbol{g}}L_{\boldsymbol{f}}^{k-1}b(x)=0$ 
for all \(k < r\). Here, \(L_{\boldsymbol{g}}L_{\boldsymbol{f}}\) denotes the Lie derivative notation~\cite{khalil2015nonlinear}.

If the system is a first-order control linear system, meaning its relative degree is one, then the problem can be solved using the CBFs.
Safety and task satisfaction are coded into barrier functions and solved for the control input.
However, if the relative degree of the system is greater than one,
then~\eqref{eq: cbf} is no longer applicable since ${u}$ does not appear in the first Lie derivative of $b({x})$. For example, in the classic adaptive cruise control (ACC) problem~\cite{ames2014control}, the dynamics is
\begin{equation}
    \dot v_e(t) = u(t), \quad \dot d(t) = v_0 - v_e(t),
\end{equation}
where $v_e(t)$ is the velocity of the ego vehicle and $d(t)$ is the distance between the ego vehicle and the preceding vehicle which maintains a constant moving speed $v_0$. Let $x(t) = [d(t), v_e(t)]^\top$,
we construct the barrier function $b(x,t) = d(t) - d_{\delta}$ to ensure that the distance is greater than $d_{\delta}$ for all times.
After applying~\eqref{eq: cbf}, the term $\frac{\partial b(x,t)^\top}{\partial x}g(x)u$ becomes 0. 
Thus, we cannot use CBFs to formulate an optimization control problem.

While higher-order control barrier functions (HOCBFs) can construct the forward invariant set for systems with higher relative-degree systems\cite{xiao2021high}. 
However, they are overly conservative because it takes multiple times for the derivatives to incorporate control input into the safety constraints~\cite{xiao2023barriernet}. 
Thus, it is difficult to apply control barrier functions to systems with a high relative degree and restricted safe sets such as when obstacles are closely spaced. 

In this paper, we consider the problem of controlling a high-order system to satisfy a specification given as an STL formula and remain safe during task completion. 

\begin{problem} Given a feedback-linearlizable system, an STL specification $\phi_{stl}$ in an environment with obstacles. Find the controller such that the trajectory of the system $x(t)$ satisfies~\eqref{eq: stl_obs}.
\label{pb:hosys}
\end{problem}

\section{Solution}
In section~\ref{sec: A}, we first introduce the ERG-guided control barrier functions to solve navigation for STL task satisfaction.
A reference governor is constructed as a first-order system that is directly applied with the first-order CBFs for navigation. The agent as a high-order control system tracks the governor via a stable controller with the safety guarantees
Then in section~\ref{sec: B} we apply differentiable programming to iteratively learn the control parameters and improve the performance by reducing the STL task completion time.



\subsection{ERG-guided CBF}
\label{sec: A}

Consider the dynamic feedback-linearizable system with a high relative degree in~\eqref{eq: linear_sys}. The output of the system $y = Cx$ is set to track a reference governor $g(t) \in \mathbb{R}^p$. 
The system admits the controller ${u} = {K(x - \bar x_g)}$ such that the closed-loop system

\begin{equation}
\label{eq: close_sys}
\begin{aligned}
    \dot{{x}}&={Ax}+{BK({x} - \bar x_g)}\\
    y &= Cx,
\end{aligned}
\end{equation}
is stable for the equilibrium at the point $\bar x_g$ if the matrix $({A+BK})$ is Hurwitz\cite{francis1977linear}.

The dynamic of the reference governor is constructed as a simple first-order linear system:
\begin{equation}
\label{eq: rg}
    \dot g = u_g
\end{equation}
where $u_g$ is the control input for the reference governor.

For controllable ${(A , B)}$,
the energy of the system in~\eqref{eq: linear_sys} is 
\begin{equation}
    V({x}, \bar x_{g})=({x}-\bar x_{g})^{\top} {P}({x}-\bar x_{g})
    \label{eq :ly}
\end{equation}
where ${P}$ is the unique solution of the Lyapunov equation  ${(A+B K})^{\top} {P}+{P}({A}+{B K})=-{Q}$ for any positive-definite symmetric matrix ${Q}$. Denote the energy function as $\left\|x - \bar x_g\right\|_P^2$.

\begin{lemma}
\label{lemma: l}
Consider the output ${y = Cx}$. The value of the Lyapunov function in~\eqref{eq :ly} is such that
\begin{equation}
\label{eq: lemma1}
    \left\|{C} {x}-{g}\right\|^2 \leq l^2\left\|x - \bar x_g\right\|_P^2
\end{equation}
where $l = \lambda_{max} (L^{-1} C^\top CL^{-\top} )$, and $L$ is the square root of positive-definite matrix $P$, i.e., $P = LL^\top$.
\begin{proof}
The proof follows from the bound of the Rayleigh
quotient $R(A,z) = \frac{z^\top Az}{z^\top z} \leq \lambda_{max}(A)$.
Let $z = L^\top({x}_1 - {x}_2)$ and $A = L^{-1}C^\top C L^{-\top}$. Then the inequality can be rewritten as:
\begin{equation*}
z^\top A z \leq \lambda_{max}(A) z^\top z,
\end{equation*}
which leads to
\begin{equation*}
({x}_1 - {x}_2)^\top L A L^\top ({x}_1 - {x}_2) \leq \lambda_{\text{max}}(A) ({x}_1 - {x}_2)^\top P ({x}_1 - {x}_2).
\end{equation*}

or equivalently,
\begin{equation*}
\|C({x}_1 - {x}_2)\|^2 \leq \lambda_{\text{max}}(A) \|{x}_1 - {x}_2\|_P^2.
\end{equation*}
\end{proof}
\end{lemma}
Define the distance between the governor reference ${g}$ to the nearest obstacle as $d_s(g,\mathcal{O}) \in \mathbb{R}$. 

\begin{prop}
\label{prop: lyp}For a fixed $g\in \mathcal{F}$, 
    $\Delta(x,g) = d^2_s({g},\mathcal{O}) - l^2V(x,g)$ is a barrier function, where $V(x,g)$ is the Lyapunov function in~\eqref{eq :ly}, and $l$ is defined in Lemma~\ref{lemma: l}. The set $\{x\mid \Delta(x,g)\geq 0\}$ is positively forward invariant, the output $y(t)$ converges to $g$ asymptotically and $y(t) \in \mathcal{F}$.
\end{prop}
\begin{proof}
    If the initial value $\Delta(x(0), g) >0$ and since $V(x,g)$ is a Lyapunov function, the time derivative of $\Delta(x,g) = \frac{\partial \Delta(x,g)}{\partial x}\dot x = -l^2\frac{\partial V(x,g)}{\partial x}\dot x > 0$. Hence, the set $\{x\mid \Delta(x,g)>0\}$ is forward invariant. The controller in~\eqref{eq: close_sys} guarantees the convergence of output tracking.
\end{proof}

\begin{lemma}
    $\Delta(x,g)$ is a valid dynamic safety margin for the closed-loop system in~\eqref{eq: close_sys}.
    \label{lemma: dsm}
\end{lemma}
\begin{proof}
Consider a positive dynamic safety margin. We have
\begin{equation}
\begin{aligned}
\Delta({x}, g) \geq 0 & \Longrightarrow d_s^2(g, \mathcal{O}) \geq l^2\|{x}-{\bar x_g}\|_{P}^2 \\
& \Longrightarrow d_s({g}, \mathcal{O}) \geq\|{C} {x}-{g}\|
\end{aligned}
\end{equation}
Thus $g\in\mathcal{F}$ implies $y = Cx \in cl(\mathcal{F})$. Using Prop.~\ref{prop: lyp}, the four conditions in Def.~\ref{def: DSM} can be proved; see~\cite{li2023governor} for details.
\end{proof}

The dynamic safety margin $\Delta(x,g)$ specifies how safe the governor's location is.
The navigation field is used as a direction change for the governor's state.
One way is to construct artificial potential fields that are designed to satisfy Def.~\ref{def: NV}. However, artificial potential fields are known to have some limitations such as the inability to pass between closely spaced obstacles, oscillation between obstacles, and getting stuck in local minima~\cite{koren1991potential}.
This paper focuses on satisfying the STL specification that can be leveraged to formulate the navigation field as the objective of ERG in Thm.~\ref{thm: erg} especially,


\begin{definition}
\label{def: NV2}
A function $\rho(g): \mathbb{R}^p  \rightarrow \mathbb{R}^p$ is a navigation field if for any $g(0)$ satisfying the constraints~\eqref{eq: constr}, the system $\dot g = \rho(g)$ is such that
\begin{enumerate}
    \item $\sup _{(g) \in H}\|\rho(g)\|$ is finite for each compact set $H$.
    \item For any continuous reference $g(t) \in \mathbb{R}^p$, the resulting $g(t)$ satisfies $c\left(\bar{x}_g ,g(t)\right) \geq \delta$.
\end{enumerate}    
\end{definition}

The reference governor defined in~\eqref{eq: rg} is a first-order linear system. Therefore the obstacle navigation for the governor can be solved by constructing the control barrier functions as a quadratic programming problem:
\begin{subequations}
\label{eq: cbf2}
\begin{align}
    &\text{min }  \qquad u_g^\top H u_g \label{a}\\
      &\text{s.t.}\:  \frac{\partial {b}_{obs}(g)}{\partial g}u_g \geq-\alpha({b}_{obs}(g)) \label{b} \\
    & \frac{\partial {b}_{stl}(g, t)^T}{\partial g}\Delta(t)u_g+\frac{\partial {b}_{stl}(g, t)}{\partial t} \geq-\alpha({b}_{stl}(g, t)) \label{c}\notag\\
\end{align}
\end{subequations}

where $H\in \mathbb{R}^{m\times m}$ is a positive semi-definite matrix, ${b}_{stl}$ and ${b}_{obs}$ are the corresponding control barrier functions for the STL formula~\cite{lindemann2018control} and obstacle avoidance, $\alpha$ is a class K function and $\Delta(t)$ is the value of DSM at time $t$. 

\begin{prop}\label{prop: nv}
    The controller in~\eqref{a} is a valid navigation field for Def.~\ref{def: NV2}. 
\end{prop}
\begin{proof}
The control $u_g$ can be directly bounded through optimization constraints. If $g(0) \in \mathcal{F}$ and~\eqref{b} is feasible, then $\{g\mid {b}_{obs}(g) \geq 0\}$ is a forward invariant set which means $g(t) \in \mathcal{F}$. Since $\bar x_g $ is the equilibrium point from $g(t)$ to the space of $x(t)$, then $g(t) \in \mathcal{F}$ iff $\bar x_g(t) \in \mathcal{F}_x$ satisfies Def.~\ref{def: NV2}.
\end{proof}


\begin{theorem} 
\label{thm: erg_stl}
    Consider the prestabilized system in~\eqref{eq: close_sys} and constraints in~\eqref{eq: constr} using the navigation field and the dynamic safety margin in the Lemma~\ref{lemma: dsm} and Prop.~\ref{prop: nv}. Given the initial condition $x(0), g(0)$ such that $c(x(0), g(0)) > 0$, the controller
    \begin{equation}
    \dot g = \Delta(x,g)u_g,
\label{eq: au_ref2}
\end{equation}
satisfies constraints~\eqref{eq: constr} at all times for any piecewise continuous reference signal $g(t) \in \mathbb{R}^p$.

Then. the governor trajectory $g(t)$ is guaranteed to satisfy the STL formula, and the system output trajectory $y$ converges to $g$ and $x$ converges to $\bar x_g$.

\end{theorem}


\begin{proof}
The proof is based on the proof for Thm.~\ref{thm: erg}~\cite{nicotra2018explicit}.
    Since $\dot g$ is finite, $g(t)$ exists and is continuous ~\cite{filippov2013differential}.
    Likewise, since system~\eqref{eq: linear_sys} is Lipschitz, the signal of $x(t)$ is also continuous.
    If the initial condition satisfies the constraints Def.~\ref{def: DSM}, $\Delta(0) > 0$. 
    From continuity, we have that if there exists a time $t$ such that $\Delta(t) < 0$, there must be a time $t^* < t$ such that $\Delta(t^*)=0$. 
    However, since $\Delta(t^*)=0$ implies $\dot g(t^*)=0$ from~\eqref{eq: au_ref}. Therefore, since $\Delta$ is a valid DSM, by Def.\ref{def: DSM}, $\Delta(t^* + T)$ is nonnegative for $T \geq 0$, which leads to a contradiction to the $\Delta(t) < 0$. Thus,~\eqref{eq: constr} is satisfied.
    
     The STL satisfaction for the governor can be guaranteed by using the STL-CBF constraints in~\eqref{c}. The convergence property is proved in Prop.~\ref{prop: lyp}.
\end{proof}

When the governor is close to the obstacle, the DSM is smaller, which slows down the update. Therefore, we add a distance term to the objective function~\eqref{a}
    \begin{equation}
    \label{eq: obj_u}
        \text{min } u_g^\top H u_g + d^\top Qd
    \end{equation}
    where $d =d(x(t,u_g),\mathcal{O})$ and $Q$ is a negative definite matrix such that the governor maintains a feasible distance from obstacles while still satisfying the constraints.

\begin{remark}
    The barrier function for obstacle avoidance $b_{obs}$ in~\eqref{eq: cbf2} can be also coded as part of the STL formula using conjunction as in~\eqref{eq: cbf_and}. However, the construction of CBFs for the STL formula tends to be conservative and involves handpicked parameters and structure. To demonstrate the effectiveness of the reference governor approach, we use independent constraints for obstacle avoidance in this paper. 
\end{remark}

\subsection{Iterative Tuning}
\label{sec: B}
The key component of satisfying STL specifications is the performance of tracking controllers of the agents in~\eqref{eq: close_sys}. To improve it, we employ differentiable programming and iteratively improve the control parameters.

Let the task completion times for the governor and the agent be denoted as $t_g$ and $t_a$, respectively. Thm.~\ref{thm: erg_stl} ensures the safe tracking of the governor $g(t)$ by the agent $y(t)$. Additionally, it guarantees that the governor complies with the STL formula for $g(t)$ and the agent trajectory $x(t)$ eventually converges to $\bar x_g(t)$. However, the exact point in time, $t_a$, when the agent complies with the STL formula, is not necessarily restricted within the time windows defined for the STL specifications. This is largely dependent on the parameters of the controller. 

In Prop.~\ref{prop: lyp}, the DSM is constructed based on a heuristic feedback controller to stabilize the system. Here, we apply auto-differential iterative tuning to improve the performance of the parameters in the feedback controller, thus minimizing the tracking time and, as a result, decreasing $t_a$.

Iterative tuning methods involve iteratively updating the parameters for evaluations to improve performance based on a loss function (e.g., tracking error) often using gradient-based approaches~\cite{berkenkamp2016safe,cheng2023difftune}.

In our settings, we apply a model-based tuning method called \emph{DiffTune}~\cite{cheng2023difftune} which uses the sensitivity equation to propagate the gradient. 
Denote the parameters of the closed-loop controller as $\boldsymbol{\theta}$.
The loss is the tracking error between the agent and the governor over the task completion time $t_a$:
\begin{equation}
\label{eq: loss}
        \mathcal{L}: \sum_{t=0}^{t_a
}(Cx(t) - g(t)),
\end{equation}
The parameters $\boldsymbol{\theta}$ are updated as:
\begin{equation}
\boldsymbol{\theta} \leftarrow \boldsymbol{\theta}-\alpha \nabla_{\boldsymbol{\theta}} \mathcal{L},
\end{equation}
where $\alpha$ is the step size and $\nabla_{\boldsymbol{\theta}} \mathcal{L}$ is the gradient from the sensitively function~\cite{khalil2015nonlinear}. Thus, $\boldsymbol{\theta}$ is iteratively updated to decrease the loss and minimize the total tracking time.

\section{Simulation Results}
In this section, we assess the performance of the ERG-guided CBF for high-order systems with STL specifications. We show two case studies for our evaluation. The first case uses a double integrator model. The second case uses the quadrotor model showing the application for the feedback-linearizable system.
\vspace{-5pt}
\subsection{Double integrator model}
\subsubsection{Specifcations}

Consider the agent dynamics as a double integrator. The reference governor is a first-order governor system. Denote $x, g \in \mathbb{R}^2$ as the positions of the agent and governor in a 2D environment: 
\vspace{-5pt}
\begin{equation}
\begin{aligned}
        &\dot x = v_a, \quad  \dot v_a = u_a\\
        & \dot g = u_g,
\end{aligned}
\end{equation}
where $v_a \in \mathbb{R}^2$ is the velocity of the agent. Denote $\boldsymbol{x} = [x, \dot x]^\top$. The controller $u_a = K\boldsymbol{x}$ is a feedback controller in $\mathbb{R}^2$, where
$K = \begin{bmatrix}
    k_p& k_p& 0&0\\
    0&0&k_d&k_d
\end{bmatrix}$. The output $y = C\boldsymbol{x}$ is set to extract the position of the agent, i.e., $C = \begin{bmatrix}
    1 & 0 & 0 & 0\\
    0 & 1 & 0 & 0
\end{bmatrix}$. The matrix $K$ is initialised with $k_p = -6, k_d = -4$. The locations of the agent and governor are initialized at the origin with zero velocity.
The simulation frequency is 100 Hz, and all optimizations are solved using Gurobi~\cite{gurobi}.

The STL formula specification is 
\vspace{-5pt}
\begin{equation*}
    \levent_{[5, 30]}Reach_1 \land \levent_{[30, 80]}Reach_2 \land \lalways_{[0,80]}Stay_3,
\end{equation*}
where $Reach_i$ means the agent needs to reach a circle area {$\mathcal{R}_i$}, i.e. $\|x - \boldsymbol{o}_i\| < \boldsymbol{r}_i$, where $\boldsymbol{o}_i$ and $\boldsymbol{r}_i$ are the center and radius of area {$\mathcal{R}_i$}.
$Stay_i$ means the agent needs to stay within the circle arena area {$\mathcal{R}_3$}.
The large time windows for the subtasks are chosen to ensure the STL feasibility in the auto-tuning under different control parameters.

\subsubsection{Performance}
Fig.~\ref{fig: results_traj} shows the comparison of ERG-guided CBF and HOCBF~\cite{xiao2021high}. Obstacle locations are closely spaced within the gray arena area to create narrow passages. As shown in Fig.~\ref{fig: erg}, under the ERG-guided CBF, the agent successfully tracks the governor, completes the STL specifications and ensures safety. In contrast, the implementation of HOCBF, depicted in Fig.~\ref{fig: hocbf}, demonstrates the agent fails to reach the target. The narrow passages between obstacles prevent the successful completion of the specification thus highlighting the advantages of ERG-guided CBFs.

\begin{figure}[hbt!]
    \centering
    \subfloat[\label{fig: erg}]{%
     \includegraphics[width=0.45\linewidth]{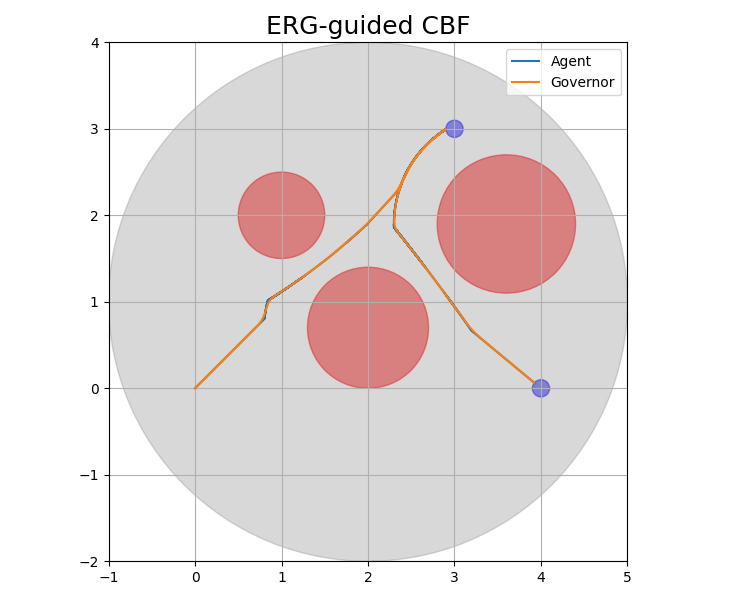}}
    \centering
    \subfloat[\label{fig: hocbf}]{%
    \includegraphics[width=0.45\linewidth]{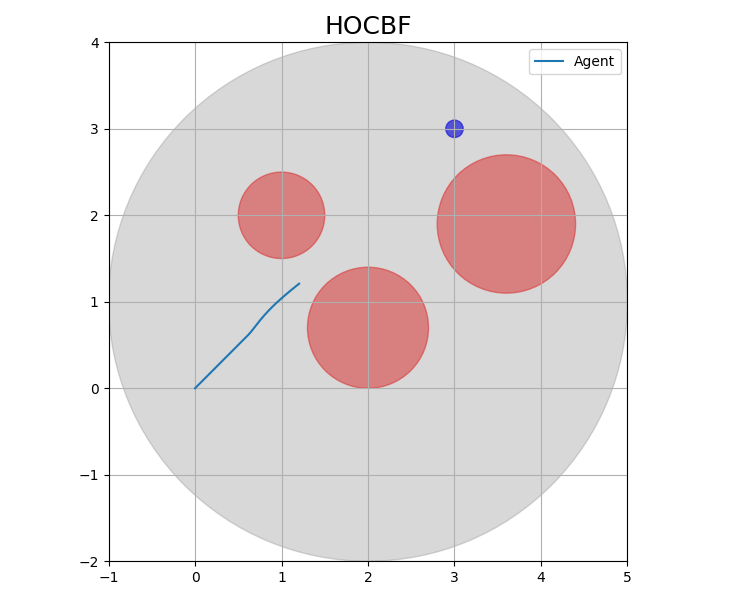}
       }
\caption{Environment: red circles are obstacles, blue circles are target areas, and grey circles are arena areas. (a) trajectories using ERG-guided CBF (b) agent trajectory using HOCBF}
     \label{fig: results_traj}
\end{figure}
\vspace{-20pt}
\begin{figure}[hbt!]
    \centering
    \begin{subfigure}[b]{0.2\textwidth}
    \includegraphics[width=\textwidth]{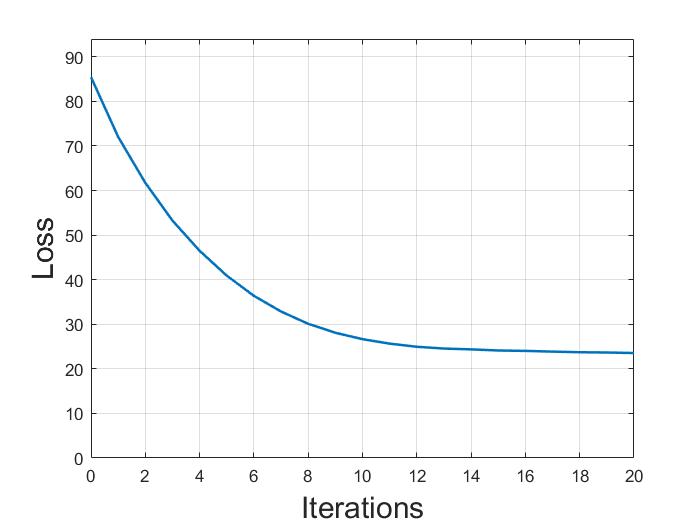}
        \caption{Loss}
    \label{fig: loss}
    \end{subfigure}
    \begin{subfigure}[b]{0.2\textwidth}
        \includegraphics[width=\textwidth]{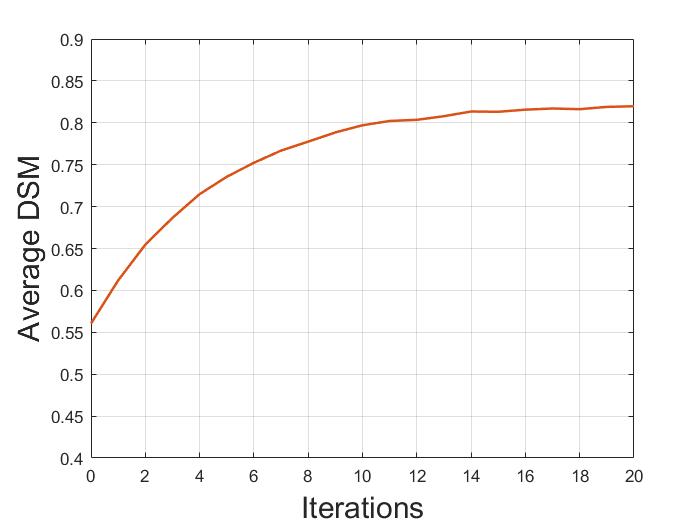}
        \caption{Average DSM}
        \label{fig: avg_dsm}
    \end{subfigure}

    \begin{subfigure}[b]{0.2\textwidth}
        \includegraphics[width=\textwidth]{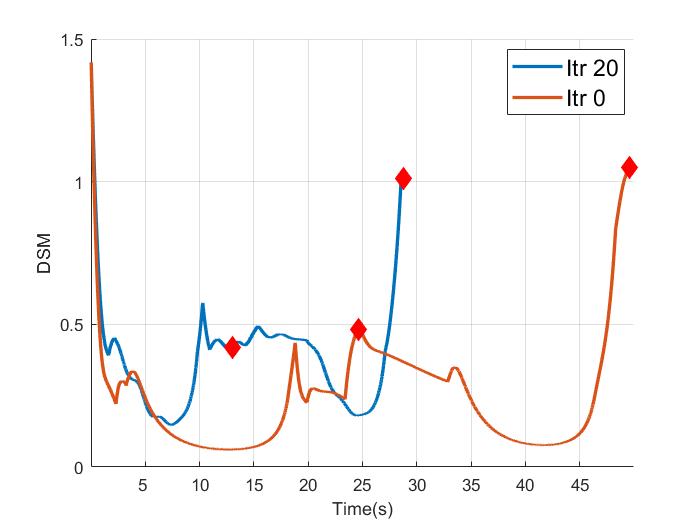}
        \caption{Task completion time}
        \label{fig: time}
    \end{subfigure}
    \begin{subfigure}[b]{0.2\textwidth}
        \includegraphics[width=\textwidth]{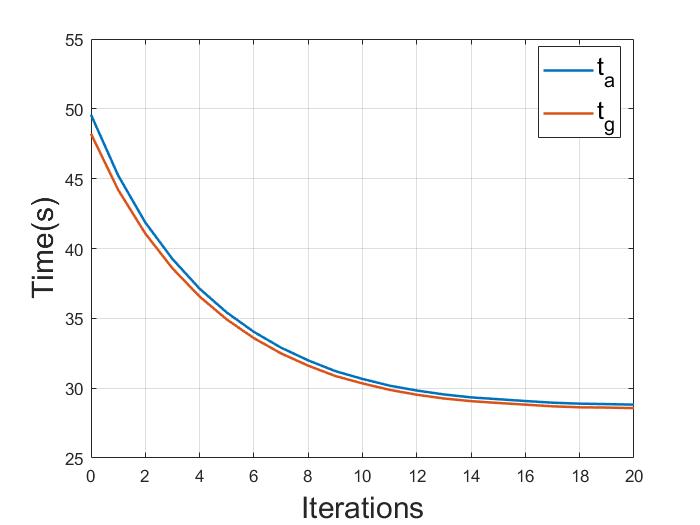}
        \caption{DSM over Time}
        \label{fig: compare_dsm}
    \end{subfigure}
    \caption{Iterative tuning performance}
    \label{fig: auto-tune}
\end{figure}
\vspace{-5pt}
Iterative tuning is used for the parameters in the feedback controller $K$. Fig.~\ref{fig: loss} illustrates the changes in the loss~\eqref{eq: loss} during 20 iterations. The results show that the tuning significantly decreases the loss through iterations. The improved control parameters, in turn, induce closer tracking of the agent to the governor, thus increasing the mean DSM in Fig.~\ref{fig: avg_dsm}. 
The governor's update process is dictated by the magnitude of the DSM. As a result, an increase in the mean DSM leads to a faster adjustment of the governor, which consequently accelerates the completion of the agent's task. Fig.~\ref{fig: time} illustrates a significant reduction in completion times for both the governor ($t_g$) and the agent ($t_a$), along with their difference ($|t_g - t_a|$) during the iterative process.  Fig.~\ref{fig: compare_dsm} compares the DSM between the initial and final iteration settings, with red diamond markers representing the time points when the agent visits the two targets. In particular, both DSM curves are above 0, indicating successful safe navigation through iterations. Moreover, in the final iteration, the agent reaches both targets and completes the STL task faster than in the initial setting and with a higher average DSM value, thus validating the effectiveness of the tuning.
\vspace{-4pt}
\subsection{Quadrotor model}
\subsubsection{Specifications}
The quadrotor model is an underactuated nonlinear system, and the dynamics are
\vspace{-3pt}
\begin{equation}
\begin{gathered}
\dot{x}=v_q,\quad m \dot{v}_q=m g e_3-f R e_3, \\
\dot{R}=R \hat{\Omega}, \quad
J \dot{\Omega}+\Omega \times J \Omega=M,
\end{gathered}
\end{equation}
\vspace{-2pt}
where $x, v_q\in \mathbb{R}^3$ is the location and velocity of the center of mass in the inertial frame, $f\in \mathbb{R}$ is the total thrust force generated by four rotors, $M\in \mathbb{R}^3$ is the total moment in the body-fixed frame, $R \in \mathrm{SO}(3)$ the rotation matrix from the body-fixed frame to the inertial frame, $\Omega \in \mathbb{R}^3$ is the angular velocity in the body-fixed frame, $m\in \mathbb{R}$ is the total mass and $J\in \mathbb{R}^{3\times 3}$ is the inertia matrix in the body-fixed frame. 

The authors in~\cite{lee2010geometric} develop a feedback-linearization model
to track a three-dimensional position and heading direction with control inputs $f, M$ chosen as
\begin{equation}
\vspace{-6pt}
\label{eq: fl-uav}
\begin{aligned}
f=- & \left(-k_x e_x-k_v e_v-m g \zeta_3+m \ddot{x}_d\right) \cdot R \zeta_3, \\
M=- & k_R e_R-k_{\Omega} e_{\Omega}+\Omega \times J \Omega \\
& -J\left(\hat{\Omega} R^T R_d \Omega_d-R^T R_d \dot{\Omega}_d\right),
\end{aligned}
\end{equation}
where $x_d(t)$ is the transnational command reference, and $e_x, e_v, e_R$ and $e_{\Omega}$ are the tracking errors between $x$ and $x_d$, $\zeta_3 = [0, 0, 1]^\top$.
Using feedback linearization in~\eqref{eq: fl-uav}, we obtain $\ddot x_d = u_d$ such that the position control for the nonlinear dynamics of the drone is simplified to controlling $x_d$.

The feedback linearization regulates the center of mass position of the quadrotor, and to ensure the entire frame of the quadrotor is safe, we inflate the obstacles by the distance from the center of mass to a rotor.
The STL specification is
\vspace{-3pt}
\begin{equation*}
    \levent_{[5, 30]}Reach_1 \land \levent_{[30, 50]}Reach_2.
\end{equation*}
\subsubsection{Performance}
Fig.~\ref{fig: drone_traj} shows the results from the 20th iteration. The figure on the left shows the navigation trajectory of the quadrotor through obstacles starting from the origin to a way point and then to the destination. The figure on the right shows the DSM over time, which remains positive.


\begin{figure}[hbt!]
    \centering
    \subfloat[\label{fig: erg_drone}]{%
     \includegraphics[width=0.5\linewidth]{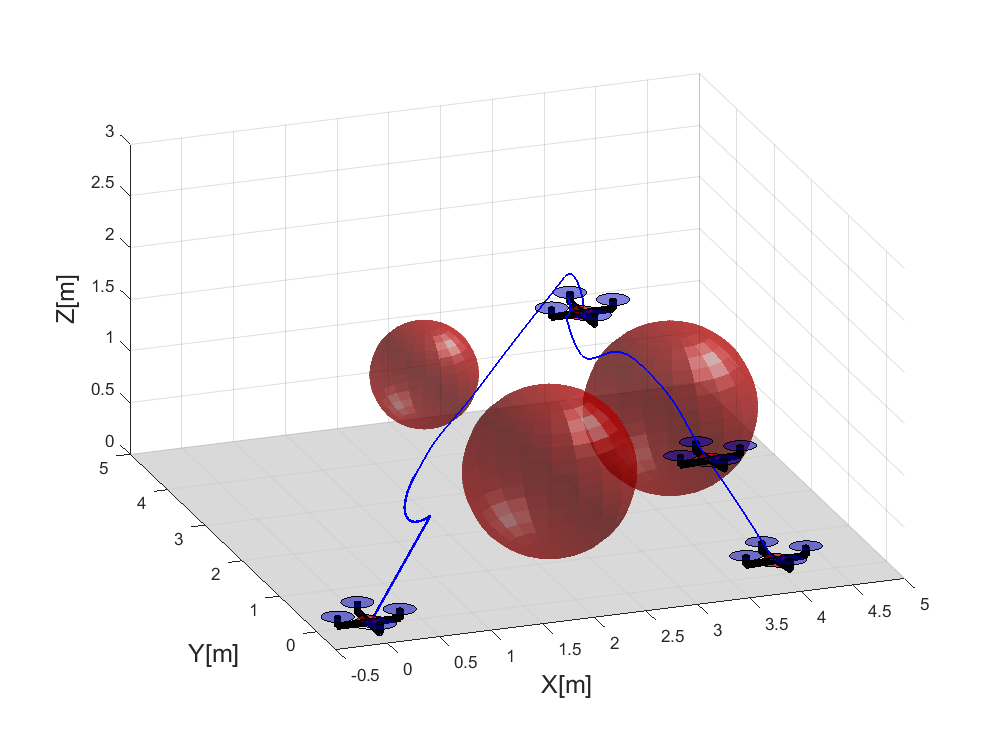}}
    \centering
    \subfloat[\label{fig: dsm_drone}]{%
    \includegraphics[width=0.45\linewidth]{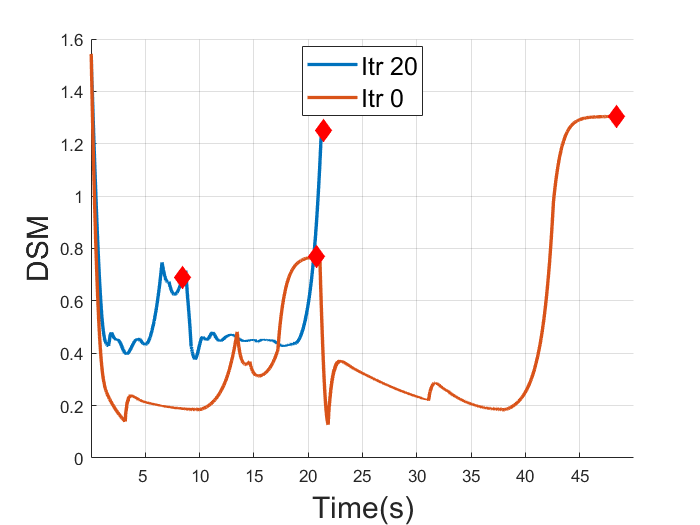}
       }
     \caption{ (a) Quadrotor trajectory. (b) DSM over time.}
    \label{fig: drone_traj}
\end{figure}

\vspace{-7pt}
\section{Conclusion}
This paper develops ERG-guided CBFs that assure safety for high-order linearizable systems with STL tasks. Our approach demonstrates that by employing the explicit reference governor, we can leverage first-order CBFs to manage a system with a high relative degree. Furthermore, the controller for such high-order systems can be optimized using gradient-based methods via iterative tuning, thus enhancing the performance of the CBFs. 
\vspace{-5pt}
\bibliographystyle{ieeetr}
\bibliography{references.bib}
\end{document}